\documentclass{llncs}
\usepackage{amsmath,amsfonts}
\usepackage{ae,aecompl}
\usepackage{enumerate}
\usepackage{graphicx}
\usepackage{color}
\usepackage{verbatim}

\newcommand{\COMMENTED}[1]{}

\newcommand{\ADV}{\textsc{Adv}}

\newcommand{\E}{\mathbf{E}}

\newcommand{\RMIX}{\textsc{RMix}}
\newcommand{\e}{\mathrm{e}}

\title{Randomised Buffer Management with Bounded Delay against Adaptive Adversary}

\author{{\L}ukasz Je{\.z}\inst{1}}

\institute{
Institute of Computer Science,
University of Wroc{\l}aw,
50-383 Wroc{\l}aw, Poland.
}

\begin{document}

\maketitle

\section{Introduction}

We study the Buffer Management with Bounded Delay problem, introduced by Kesselman~et~al.~\cite{DBLP:journals/siamcomp/KesselmanLMPSS04},
or, in the standard scheduling terminology, the problem of online scheduling of unit jobs to maximise weighted throughput. The adaptive-online
adversary model for this problem has recently been studied by Bie{\'n}kowski~et~al.~\cite{DBLP:conf/waoa/BienkowskiCJ08}, who proved a lower
bound of \(\frac{4}{3}\) on the competitive ratio and provided a matching upper bound for \(2\)-bounded sequences. In particular, the authors
of~\cite{DBLP:conf/waoa/BienkowskiCJ08} claim that the algorithm $\RMIX$~\cite{DBLP:journals/jda/ChinCFJST06} is \(\frac{\e}{\e-1}\)-competitive
against an adaptive-online adversary. However, the original proof of Chin~et~al.~\cite{DBLP:journals/jda/ChinCFJST06} holds only in the oblivious
adversary model. The reason is as follows. First, the potential function used in the proof depends on the adversary's future schedule, and second,
it assumes that the adversary follows the earliest-deadline-first policy. Both of these cannot be assumed in adaptive-online adversary model,
as the whole schedule of such adversary depends on the random choices of the algorithm. We give an alternative proof that {\RMIX} indeed is
\(\frac{\e}{\e-1}\)-competitive against an adaptive-online adversary.

Similar claim about {\RMIX} was made in another paper by Bie{\'n}kowski~et~al.~\cite{DBLP:conf/soda/BienkowskiCDHJJS09} studying a slightly
more general problem. It assumes that the algorithm does not know exact deadlines of the packets, and instead knows only the order of their
expirations. However, any prefix of the deadline-ordered sequence of packets can expire in every step. The new proof that we provide holds even
in this more general model, as both the algorithm and its analysis rely only on the relative order of packets' deadlines.

\section{{\RMIX} and its new analysis}

The algorithm $\RMIX$ works as follows.
In each step, let $h$ be the heaviest pending job. 
Select a real $x \in [-1,0]$ uniformly at random. Transmit $f$, the earliest-deadline 
pending packet with $w_f \geq \e^x \cdot w_h$. 

We write $a \lhd b$ ($a \unlhd b$) to denote that the deadline of packet $a$ is earlier 
(not later) that the deadline of packet $b$. This is consistent with the convention
of~\cite{DBLP:conf/soda/BienkowskiCDHJJS09} for the more general problem studied therein.

\begin{theorem}
\RMIX is $\e / (\e-1)$-competitive against an adaptive-online adversary.
\end{theorem}

\begin{proof}
We use the paradigm of modifying the adversary's buffer used in the paper of Li et al.~\cite{DBLP:conf/soda/LiSS05}.
Namely, in each time step we assume that $\RMIX$ and the adversary $\ADV$ have the same buffers. 
Both $\RMIX$ and $\ADV$ transmit a packet. If after doing so, the contents of their buffers become 
different, we modify the adversary's buffer to make it identical with that of $\RMIX$.
To do so, we may have to let the adversary transmit another packet and keep the one originally transmitted in the buffer,
or upgrade one of the packets in its buffer by increasing its weight and deadline. We show that in each step 
the expected gain of $\RMIX$ is at least $\frac{\e-1}{\e}$ times the expected {\em amortized gain} of the adversary, 
denoted $\ADV'$. The latter is defined as the sum of weights of the packets that $\ADV$ eventually transmitted in the step.
Both expected values are taken over possible random choices of $\RMIX$.

First, we compute the expected gain of $\RMIX$ in a single step. 
\[
	\E[\RMIX] = \E[w_f] = \int_{-1}^0 w_f \; dx 
	\enspace.
\]

Assume now that $\ADV$ transmits a packet $j$.
Without loss of generality, we may assume that for each packet $k$ from the buffer, either $w_j \geq w_k$ or 
$j \unlhd k$. We call it a {\em greediness property}.
We consider two cases. 
\begin{enumerate}
\item $f \lhd j$. By the greediness property, $w_j \geq w_f$. After both $\ADV$ and $\RMIX$ transmit their 
packets, we replace $f$ in the buffer of $\ADV$ by $j$.
\item $j \unlhd f$. After both $\ADV$ and $\RMIX$ transmit their packets, we let $\ADV$ transmit additionally 
$f$ in this round and we reinsert $j$ into its buffer.
\end{enumerate}
Therefore the amortized gain of $\ADV$ is $w_j$ and additionally $w_f$ if $j \unlhd f$.  
By the definition of the algorithm, $j \unlhd f$ only if $w_f \geq w_j$. 
Let $y = \ln (w_j/w_h)$. Then,
\[
	\E[\ADV'] = w_j + \E[w_f | w_f \geq w_j] =
 		w_j + \int_{y}^0 w_f \; dx
	\enspace.
\]
Finally, we compare the gains, obtaining
\begin{align*}
\frac{\E[\RMIX]}{\E[\ADV']} \;= &\; 
	\frac{\int_{-1}^y w_f \; dx + \int_y^0 w_f \; dx}{w_j + \int_{y}^0 w_f \; dx} \geq 
	\frac{\int_{-1}^y \e^x w_h \; dx + \int_y^0 \e^x w_h \; dx}{w_j + \int_{y}^0 \e^x w_h \; dx} \\
	= &\; \frac{\int_{-1}^0 \e^x w_h \; dx}{w_j + \int_{y}^0 \e^x w_h \; dx} 
		= \frac{w_h \cdot (1-1/\e)}{w_j + w_h \cdot (1-w_j/w_h)} \\
	= &\; 1 - 1/\e 
\enspace,
\end{align*}
which concludes the proof.
\qed
\end{proof}

\bibliographystyle{abbrv}
\bibliography{online}

\end{document}